\documentclass[a4paper,11pt]{article}
\pdfoutput=1

\usepackage{doc-jmr}
\usepackage{bm}
\usepackage{localmacros}

\title{Efficient capacity-achieving entanglement generation with application to pure-loss Bosonic channels}
\author[1,2]{Avijit Mandal\footnote{avijit.mandal@duke.edu}}
\author[4]{Christophe Piveteau}
\author[1,2,3]{Henry D. Pfister}
\author[5]{Joseph M. Renes}
\affil[1]{Department of Electrical and Computer Engineering, Duke University}
\affil[2]{Duke Quantum Center, Duke University}
\affil[3]{Department of Mathematics, Duke University}
\affil[4]{IBM Research Europe -- Zurich}
\affil[5]{Institute for Theoretical Physics, ETH Zurich}

\date{}
\begin{document}

\maketitle

\begin{abstract}
While the ultimate limits on the rate of quantum communication over noisy channels are well-established, very few efficient coding schemes achieving these limits are known. 
This contrasts sharply with transmitting classical data over classical channels, where polar codes and low-density parity-check codes both efficiently achieve the classical capacity of any noisy channel. 
In this work we construct an efficient entanglement generation scheme which achieves the coherent information for a broad class of channels. 
For degradable channels in the class, which includes truncated versions of the pure loss Bosonic channel, the resulting scheme achieves the quantum capacity. 
The construction makes use of polar codes for two kinds of classical-input, quantum-output (CQ) channels: those with commuting outputs and those whose outputs are heralded mixtures of pure states. 
The former case is amenable to decoding by classical polar decoding algorithms, the latter by the belief propagation with quantum messages (BPQM) algorithm. 
To our knowledge, this construction is the first explicit and efficiently decodable protocol that achieves the energy-constrained quantum capacity of the pure-loss Bosonic channel.
\end{abstract}

\section{Introduction}
The notion of quantum communication can be traced back to the ``conjugate coding'' scheme of Wiesner~\cite{wiesner_conjugate_1983} in the late 1970s, which demonstrated that transmitting {quantum} states from point to point could be useful, and that quantum effects were not just a source of noise. 
This naturally raised the question of the ultimate limits of quantum information transmission over noisy channels, which was settled by the converse bounds of Barnum, Knill, Nielsen, and Schumacher~\cite{schumacher_quantum_1996,barnum_information_1998,barnum_quantum_2000} matching random coding achievability bounds by Lloyd~\cite{lloyd_capacity_1997}, Shor~\cite{shor_quantum_2002}, and Devetak~\cite{devetak_private_2005}. 
These works established that the ultimate communication rate of a noisy channel is its regularized coherent information. 

The question of achieving these limits or even just the coherent information itself by efficient means is very much still open.
In particular, the major difficulty is in constructing a suitable decoder. 
There is significantly more progress in showing that particular codes can achieve high rates when decoded by the optimal or a near-optimal decoder for specific channels, but efficiency of the decoding operation is most often left open. 
Indeed, by the decoupling approach, one can establish that a code is good without analyzing the decoder at all~\cite{hayden_decoupling_2008}. 
Nevertheless, several constructions with both efficient encoders and decoders that can operate at rates up to the coherent information are known for specific channels: Reed-Muller codes for erasure channels~\cite{kumar_reed-muller_2016}, quantum polar codes for erasure and Pauli channels~\cite{renes_efficient_2012,dupuis_polarization_2021} and the qubit amplitude damping channel~\cite{renes_belief_2017-3,piveteau_quantum_2022}, low-depth Clifford circuits and tensor network decoding of Pauli channels~\cite{darmawan_lowdepth_2024,liu_approximate_2026}, and Bosonic Gaussian displacement channels using polar codes concatenated with Gottesman-Kitaev-Preskill (GKP) codes~\cite{subramanian_achievable_2025}. 

In this paper, we construct efficient entanglement generation schemes which achieve the coherent information for a class of channels we refer to as pure discrete phase covariant (PDPC) channels. 
These channels include amplitude damping and dephasing channels of any dimension, resonant multilevel amplitude damping (ReMAD) channels~\cite{chessa_resonant_2023}, as well as truncated versions of the pure loss and pure amplification Bosonic channels. 
Our construction nominally makes use of entanglement assistance, though when the PDPC channel is degradable, this can be removed. 
We consider entanglement generation instead of communication, though they are closely related. 
Instead of accepting an input to be transmitted, the generation task is simply to create high-fidelity entanglement between sender and receiver. 
It is known that the entanglement generating capacity of a channel is equal to its communication capacity~\cite[\S4.5 \& \S6.1]{kretschmann_tema_2004}.

Our construction is not unrelated to Wiesner's conjugate coding in that, as in~\cite{renes_physical_2008,renes_efficient_2012,renes_polar_2014,renes_uncertainty_2016}, the quantum decoder is constructed from two POVMs designed to decode classical information transmitted over the channel in one of two conjugate bases. 
The quantum decoder thus inherits the circuit complexity of the two POVMs. 
For PDPC channels and suitable codes such as polar codes, the corresponding decoding POVMs can be performed efficiently using classical polar decoders~\cite{arikan_channel_2009,arikan_source_2010} and the belief propagation with quantum messages (BPQM) algorithm~\cite{renes_belief_2017-3,rengaswamy_belief_2021,pfister_belief_2022,piveteau_quantum_2022,piveteau_efficient_2025}. 
The resulting entanglement generation protocol will achieve high fidelity outputs at rates up to the coherent information. 
Moreover, the required amount of entanglement assistance will vanish when using polar codes~\cite{sutter_universal_2014,renes_alignment_2016}.

\subsection{Technical overview}
Before diving into the details, let us give a higher-level technical overview.
We first show the construction for general entanglement-assisted CSS (EACSS) codes, general channels, and generic POVMs and then discuss efficiency when using polar codes and specific decoders for PDPC channels. 
For the purposes of entanglement generation, the sender Alice prepares $m$ i.i.d.\ copies of a pure bipartite state $\ket{\chi}_{AB}$ with $\text{dim}(A)=d$ for prime $d$ and sends half via the noisy channel to the receiver Bob. 
They now share $\sigma_{AB}^{\otimes m}$ for some state $\sigma_{AB}$ and proceed with entanglement distillation. 

To this end, Alice applies a syndrome extraction circuit $V$ for an EACSS code. 
The outputs corresponding to the commuting stabilizers are measured and the results transmitted to Bob, while the outputs corresponding to the noncommuting stabilizers must be teleported to Bob. 
The results of the stabilizer measurements give $X$- and $Z$-type syndromes, though of course the noncommuting stabilizers cannot be measured simultaneously. 

Next, Bob applies a quantum decoder constructed from two POVMs. 
The first POVM $\Lambda$ is designed to recover the value of $z^m$, the result of a hypothetical measurement on $A^m$ in the standard basis, from the syndrome and the $B^m$. More concretely, the POVM needs to distinguish the states  $\rho_{B^m}(z^m)=\rho_{B_1}(z_1)\otimes \cdots \rho_{B_m}(z_m)$, where $\rho_B(z)$ is the normalized version of $\bra{z}_A\sigma_{AB}\ket{z}_A$, when the syndrome of $z^m$ is known. 
This operation is performed coherently, storing the result in an ancilla register $A'^m$, with $A'\simeq A$, and simply controlling the operation on the standard basis of the teleported systems. 

The task of the second POVM $\Gamma$ is to determine $x^m$ from the states $\sigma_{A'^mB^m}(x^m)$ for $\sigma_{A'B}(x)=Z_{A'}^{x}\sigma_{A'B}Z_{A'}^{-x}$ and $Z_{A'}$ the generalized Pauli $Z$ operator, when the syndrome of $x^m$ is known. 
Again this operation is performed coherently, though now on Bob's $A'^mB^m$ registers, and the result is stored in a second ancilla. 
The reason this measurement and these states are relevant is the following. 
We can regard $x^m$ as the result of a hypothetical measurement on $A^m$ in the Fourier conjugate basis. 
But since the first stage is already complete, assuming that $z^m$ was properly coherently recovered, Bob now effectively has a coherent copy of the standard basis of $A$ in his system $A'$. Working out the marginal state of $A'B$ given outcome $x$ leads to $\sigma_{AB}(x)$. 

Once both POVMs have been coherently implemented, to obtain a maximally-entangled state in Alice's qubits corresponding to the logical information and Bob's second ancilla, it only remains to remove some phases which depend on the logical information by a suitable controlled-phase operation. 
The entire distillation operation is depicted in Figure~\ref{fig:entanglementgeneration}. 
Supposing that the first POVM has an average error (averaging over the $z^m$ values) of $\eps_Z$ and the second an average error of $\eps_X$, Theorem~\ref{thm:edcircuit} establishes that the fidelity of the output with the entangled state is at least $1-(\sqrt{\eps_X}+\sqrt{\eps_Z})^2$. 
Meanwhile, the net rate of entanglement generation is the size of the output minus the size of the teleported qubits, i.e.\ the number of noncommuting stabilizers.

\begin{figure}
\begin{center}

\resizebox{0.75\textwidth}{!}{%
\begin{tikzpicture}[
  x=1cm,y=1cm,
  line width=0.75pt,
  line cap=round,
  font=\normalsize,
  gate/.style={draw,fill=white,minimum width=0.85cm,
               minimum height=0.65cm,inner sep=3pt},
  control/.style={circle,draw,fill=black,inner sep=0pt,
                  minimum size=3.5pt},
  target/.style={circle,draw,fill=white,inner sep=0pt,
                 minimum size=0.39cm},
  pics/discard/.style={code={
  \begin{scope}[scale=0.6, transform shape]
    \draw (0,-0.22) -- (0,0.22);
    \draw (0.11,-0.15) -- (0.11,0.15);
    \draw (0.22,-0.08) -- (0.22,0.08);
    \end{scope}
  }},
  pics/bundle/.style={code={\draw (-0.065,-0.14)--(0.065,0.14);}}
]

\def\yBar{0}
\def\yCheck{-0.70}
\def\yHat{-1.40}
\pgfmathsetmacro{\checkhaty}{(\yCheck + \yHat) / 2}
\def\yUnder{-2.1}
\def\yB{-4.10}
\def\yAnc{-5.0}
\pgfmathsetmacro{\avgy}{(\yB + \yAnc) / 2}
\def\yOut{-5.9}
\def\xLambda{5.90}
\def\xF{7.0}
\def\xGamma{7.70}
\def\xFinv{8.350}
\def\xZ{9.5}
\def\xDiscard{10.3}
\def\xEnd{11.5}

\coordinate (sourceA) at (1.750,\checkhaty);
\coordinate (sourceB) at (1.750,\yB);
\coordinate (sourceMid) at ($(sourceA)!0.5!(sourceB)$);
\draw[
    decorate, 
    decoration={brace, amplitude=5pt}, 
]  (1.4, \yB - 0.07) -- (1.4, \checkhaty + 0.07);
\node[anchor=east,inner sep=0pt] (state) at ($(sourceMid)+(-0.70,0)$)
  {$\sigma_{AB}^{\otimes m}$};
\node[gate,minimum width=0.8cm,minimum height=2.8cm]
  (V) at (4.05,\checkhaty) {$V$};
\node
  (N) at (4.05,\yB) {};

\draw (sourceA) -- (V.west)
  node[midway,above=3pt] {\small $A^m$};
\pic at ($(sourceA)!0.5!(V.west)$) {bundle};
\draw (sourceB) -- (N.west)
  node[midway,above=4pt] {\small $B^m$};
\pic at ($(sourceB)!0.5!(N.west)$) {bundle};

\draw (V.east |- 0,\yBar) -- (\xEnd,\yBar);
\foreach \yy in {\yCheck,\yHat,\yUnder}
  \draw (V.east |- 0,\yy) -- (\xDiscard,\yy);
\node[above=1pt] at (4.8,\yBar) {\scriptsize $\bar{A}$};
\node[above=1pt] at (4.8,\yCheck) {\scriptsize $\widecheck{A}$};
\node[above=1pt] at (4.8,\yHat) {\scriptsize $\widehat{A}$};
\node[above=1pt] at (4.8,\yUnder) {\scriptsize $\underline{A}$};

\draw (N.west) -- (\xDiscard,\yB);
\draw (4.7,\yAnc) node[left=5pt] {\normalsize $\ket{0}_{A'^{m}}$}
  -- (\xDiscard,\yAnc);
\draw (4.7,\yOut) node[left=5pt] {$\ket{0}_{\bar{R}}$}
  -- (\xEnd,\yOut);

\draw (\xLambda,\yCheck) -- (\xLambda,\yAnc);
\node[gate,minimum size=.80cm]
  (Lambda) at (\xLambda,\yB) {$U_\Lambda$};
\node[control] at (\xLambda,\yCheck) {};
\node[control] at (\xLambda,\yUnder) {};
\node[target] (targetLambda) at (\xLambda,\yAnc) {};
\draw[shorten >= 0.5pt] (targetLambda.west) -- (targetLambda.east)
      (targetLambda.north) -- (targetLambda.south);

\node[gate,minimum size=0.5cm]
    at (\xF,\yHat) {$F^\dagger$};

\node[gate,minimum size=0.5cm]
    at (\xF,\yUnder) {$F^\dagger$};

\node[gate,minimum size=0.5cm]
    at (\xZ,\yBar) {$F^\dagger$};
    \node[gate,minimum size=0.5cm]
    at (\xFinv,\yHat) {$F$};

\node[gate,minimum size=0.5cm]
    at (\xFinv,\yUnder) {$F$};

\draw (\xGamma,\yHat) -- (\xGamma,\yOut);
\node[gate,minimum width=0.80cm,minimum height=1.6cm]
  (Gamma) at (\xGamma,\avgy) {$U_\Gamma$};
\node[control] at (\xGamma,\yHat) {};
\node[control] at (\xGamma,\yUnder) {};
\node[target] (targetGamma) at (\xGamma,\yOut) {};
\draw[shorten >= 0.5pt] (targetGamma.west) -- (targetGamma.east)
      (targetGamma.north) -- (targetGamma.south);

\draw (\xZ,\yAnc) -- (\xZ,\yOut);
\node[gate,minimum size=0.6cm] at (\xZ,\yAnc) {$\bar{Z}$};
\node[control] at (\xZ,\yOut) {};

\foreach \yy in {\yCheck,\yHat,\yUnder,\yB,\yAnc}
  \pic at (\xDiscard,\yy) {discard};

\end{tikzpicture}
}
\end{center}
\caption{\label{fig:entanglementgeneration} Entanglement distillation scheme for input $\sigma_{AB}^{\otimes m}$ using the EACSS code with syndrome extraction circuit $V$, coherent implementations of the POVMs $\Lambda$ and $\Gamma$, and finally controlled-phase and inverse Fourier transform operations.}
\end{figure}
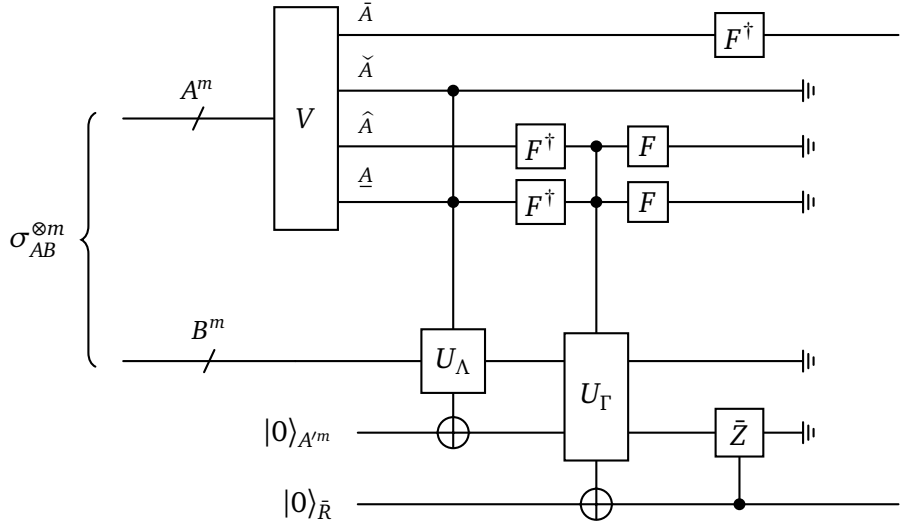

In the case of PDPC channels, the $\rho(z)$ operators all simultaneously commute, so a POVM based on a classical decoding operation is sufficient. Meanwhile, the $\sigma_{AB}(x)$ operators are unitarily equivalent to heralded pure states, specifically a pure state $Z_A^x\ket{\theta_y}_A$ for some value of $y$ which is stored in $B$. 
Employing polar codes ensures that both POVMs can be performed efficiently. The syndrome extraction circuit for any code uses at most $O(dm^2)$ controlled additions over $\ff_d$; for polar codes it uses $O(m\log m)$, since the entries of $\mathsf K_2$ are $0$ and $1$. 
The first POVM is based on the successive-cancellation decoder for polar codes~\cite{arikan_channel_2009,arikan_source_2010}, which uses $O(d^2m\log m)$ arithmetic operations. 
The second can be performed by $d$-ary BPQM~\cite{piveteau_quantum_2022,mandal_belief_2026} with $O(d^2m\log m)$ operations~\cite{piveteau_efficient_2025}, including its coherent implementation (Section~\ref{sec:implementation}).  

Moreover, they also ensure that the net rate of entanglement generation is equal to the coherent information. Polar codes achieve the information-theoretic limits, which in this case mean the required number of $Z$-type stabilizers will be roughly $mH(A|B)_\rho/\log_2 d$ (with entropy using base $2$), where $\rho_{AB}$ is the CQ state resulting from measuring $A$ in $\sigma_{AB}$ in the standard basis. 
Similarly, the required number of $X$-type stabilizers will be $mH(X|AB)_\tau/\log_2d $ for $\tau_{XAB}=\tfrac1d\sum_{x} \proj{x}_X\otimes \sigma_{AB}(x)$. 
Then the net rate of entanglement generation (measured in ebits) is $1-H(A|B)_\rho/\log_2 d-H(X|AB)_\tau/\log_2 d$, which as demonstrated in Lemma~\ref{lem:coherentinfocalc} is equal to $-H(A|B)_\sigma/\log_2 d$, the coherent information of the input state $\sigma_{AB}$. The full details are given in Theorem~\ref{thm:polar_scheme}. 

Applied to the pure loss Bosonic channel, we simply set the input state $\ket{\chi}_{AB}$ to be a truncated version of the two-mode squeezed state of some fixed average local energy $\bar n$. 
The state $\ket{\chi}$ has only an exponental loss in fidelity with the full squeezed state in the size $d$ of the truncation, and thus entropy continuity ensures that the entanglement generation rate is very close to the coherent information of the pure loss channel.

\subsection{Outline of the paper}
After giving some background including mathematical notations and some useful lemmas in the subsequent section, we give the precise details of the entanglement distillation procedure in Section~\ref{sec:entanglement_distillation}. Section~\ref{sec:polar} gives the details of the polar coding scheme, including the rate calculation. Finally, Section~\ref{sec:bpqm} discusses the efficient implementation of the POVM components of the decoder, particularly the connection to BPQM.

\section{Background}
\subsection{Mathematical preliminaries}
For a positive integer $m$, let $[m]\coloneqq\{1,\dots,m\}$. For prime $d$, $\ff_d$ denotes the finite field of integers modulo $d$. A $d$-ary string $s$ of length $\ell$ is denoted $s^\ell$ and is treated as a column vector in $\ff_d^\ell$. We often use bold symbols to denote vectors whose length is clear from context, and $\mathbf 0$ denotes the all-zero vector or the zero operator. We work with entropies and logarithms base $2$. 

For dimension $d$ set an orthonormal basis $\{\ket{z}\}_{z=0}^{d-1}$ and let $F$ be the $d$-ary Fourier transform
$F=\tfrac{1}{\sqrt{d}}\sum_{y,z=0}^{d-1}\omega^{yz}\ket{y}\bra{z}$,
where $\omega=e^{2\pi i/d}$. Define the Fourier conjugate basis $\{\ket{\widetilde x}=F\ket{x}\}_{x=0}^{d-1}$.

We denote the trace distance of two states by $\delta(\rho,\sigma)\coloneq \tfrac12\|\rho-\sigma\|_1$. 
We will also have occasion to denote the fidelity $F(\rho,\sigma)\coloneq \|\sqrt{\rho}\sqrt{\sigma}\|_1$ using `$F$', but it will be clear from the context whether this refers to the Fourier transform or the fidelity.


\subsection{Entanglement-assisted CSS codes}
Here we consider a special case; for more thorough background consult \cite{brun_correcting_2006,luo_nonbinary_2016}.
Suppose $d$ is prime and we have an $m$-qudit EACSS code $\cC$, specified as follows by non-negative integers $(k,\ell,s_x,s_z)$ with $k+\ell+s_x+s_z=m$ and an invertible matrix $H\in \ff_d^{m\times m}$. 
First define the syndrome extraction unitary $V=\sum_{z^m\in \ff_d^m} \ket{Hz^m}\bra{z^m}$, where $z^m$ in $Hz^m$ is interpreted as a column vector. 
When applied to the register $A^m$, we name the outputs of $V$ as follows. 
The first $k$ are denoted $\bar A$, the next $s_z$ $\widecheck A$, the next $s_x$ $\widehat A$, and the final $\ell$ $\underline{A}$. 
The stabilizers of the code correspond to single-system $X$- or $Z$-type operators at the output of $V$. 
The $Z$-type stabilizers correspond to single-system $Z$ operators in $\widecheck{A}$ and $\underline{A}$, while the $X$-type stabilizers correspond to single-system $X$ operators in $\widehat{A}$ and $\underline{A}$. 
Entanglement assistance is necessary to simultaneously measure all stabilizers, because those coming from $\underline{A}$ do not commute. 

The stabilizers of the code are generated by operators of the form $X^u$ and $Z^v$ for suitable $u,v\in \ff_d^m$, as follows. The $Z$-type stabilizers correspond to $V^\dagger Z^{e_j} V$ for particular $j$, where $e_j$ the length-$m$ row vector of zeros except a single 1 in the $j$th position. 
Since $V^\dagger Z^{e_j} V=Z^{e_jH}$, the possible $v$ are just the rows of $H$ which correspond to the subsets $\widecheck{A}$ and $\underline{A}$. Denote by $H_Z$ this subset of rows. 
For the $X$-type stabilizers, observe that $V$ can also be expressed in the Fourier basis using $G=(H^T)^{-1}\in \ff_d^{m\times m}$ as $V=\sum_{x^m\in \ff_d^m}\ket{\widetilde {Gx^m}}\bra{\widetilde{x^m}}$. 
An entirely similar argument then implies that the possible $u$ for the $X$-type stabilizers are just the rows of $G$ corresponding to subsets $\widehat{A}$ and $\underline{A}$. 
Denote by $G_X$ this subset of rows.

\subsection{Useful lemmas}
We will need the following coherent version of the gentle measurement lemma (see also~\cite[Exercise 9.4.2]{wilde_quantum_2017}).
\begin{lemma}[Coherent gentle measurement]
\label{lem:guessingfidelity}
Suppose $\{p_j,\varphi_B(j)\}_{j=1}^n$ is an ensemble of states and $\{\Lambda_B(j)\}_{j=1}^n$ is a POVM for which 
\begin{equation}
\sum_{j=1}^n p_j \tr[\Lambda_B(j)\varphi_B(j)]\geq 1-\eps\,.
\end{equation}
For $\ket{\varphi(j)}_{BE}$ a purification of $\varphi_B(j)$ for all $j$ and $V$ the following coherent implementation of $\Lambda$:
\begin{equation}
V:\ket{\psi}_{BE}\mapsto \sum_{j=1}^n \sqrt{\Lambda_B(j)}\ket{\psi}_{BE}\ket{j}_{A'}\,, 
\end{equation}
define the states
\begin{align}
\ket{\Psi}_{AA'BE}
&=\sum_{j=1}^n \sqrt{p_j} \ket{j}_{A}V\ket{\varphi(j)}_{BE}
=\sum_{j,j'=1}^n \sqrt{p_j} \ket{j}_A \sqrt{\Lambda_B(j')}\ket{\varphi(j)}_{BE} \ket{j'}_{A'}\\
\ket{\Psi'}_{AA'BE} 
&= \sum_{j=1}^n \sqrt{p_j} \ket{j}_A \ket{\varphi(j)}_{BE} \ket{j}_{A'}\,. 
\end{align} 
Then $\braket{\Psi'|\Psi}\geq 1-\eps$. 
\end{lemma}
\begin{proof}
The proof just uses the operator inequality $\sqrt{\Gamma}\geq \Gamma$, which holds for all $0\leq \Gamma\leq \id$:
\begin{align}
\braket{\Psi'|\Psi}
&=\sum_{j=1}^n p_k \bra{\varphi(j)}\sqrt{\Lambda(j)}\ket{\varphi(j)}\geq \sum_{j=1}^n p_k \bra{\varphi(j)}{\Lambda(j)}\ket{\varphi(j)}\geq 1-\eps\,. 
\end{align}

\end{proof}

The following gives a decomposition of the coherent information of a state, i.e.\ the negative quantum conditional entropy, into a sum of conditional entropies of CQ states. It will be useful to show that the constructed coding scheme achieves the coherent information. 
\begin{lemma}[Coherent information decomposition]
\label{lem:coherentinfocalc}
Consider an arbitrary state $\sigma_{AB}$ and define $\rho_{AB}=\sum_{k\in \ff_d} \proj{k}_A\sigma_{AB}\proj{k}_A$ and $\tau_{XAB}=\tfrac1d\sum_{x\in \ff_d}\proj{x}_X\otimes Z_A^x\sigma_{AB} Z_A^{-x}$. 
Then 
\begin{equation}
\label{eq:coherentinfocalc}
-H(A|B)_\sigma=\log_2 d-H(A|B)_\rho-H(X|AB)_{\tau}\,.
\end{equation}
\end{lemma}
\begin{proof}
By chain rules we have $H(X|AB)_{\tau}=H(AB|X)_{\tau}+H(X)_{\tau}-H(AB)_{\tau}$. 
Because the conditional states of $AB$ given $X$ are all related by unitaries, $H(AB|X)_{\tau}=H(AB)_\sigma$. 
In the next term we have $H(X)_{\tau}=\log_2 d$ by construction.
And since the average over the $Z^x_A$ operators will dephase the $A$ part of the state, $H(AB)_{\tau}=H(AB)_\rho$. 
Meanwhile, $H(A|B)_\rho=H(AB)_\rho-H(B)_\rho=H(AB)_\rho-H(B)_\sigma$, since $\rho_B=\sigma_B$. 
Substituting these equalities into the righthand side of \eqref{eq:coherentinfocalc} gives $H(B)_\sigma-H(AB)_\sigma$, the lefthand side. 
\end{proof}


\subsection{Pure discrete phase covariant channels}

Consider a channel $\cN_B$ with Stinespring dilation of the form 
\begin{equation}
\label{eq:specialchannels}
    U_{BE}:\ket{x}_{B}\otimes \ket{0}_{E}\mapsto \sum_{y\in \ff_d}\alpha_{x,y}\ket{x-y}_{B}\otimes \ket{y}_E\,,
\end{equation}
for some $\alpha_{x,y}\in \mathbb C$ such that $\sum_{y\in \ff_d}|\alpha_{x,y}|^2=1$ for all $x\in \ff_d$. 
We call such channels pure discrete phase covariant channels. 
These are a subclass of discrete phase covariant channels which satisfy, by definition, $\cN(Z\rho Z^\dagger)=Z\cN(\rho)Z^\dagger$. 
The additional property here and the reason for the prefix ``pure'' is that there is only a single state in $E$ for a given value of $y$ (here, simply the state $\ket{y}$), or equivalently a single Kraus operator per value of $y$. 
They include amplitude damping channels of any dimension, dephasing channels, ReMAD channels, as well as truncated versions of the pure loss Bosonic channel and the pure amplifier channel with suitable inputs of strictly bounded energy. 
Note that, for degradable PDPC channels, due to phase covariance and the concavity of the coherent information, the input state $\ket{\chi}_{AB}$ which maximizes $-H(A|B)_\sigma$ for $\sigma_{AB}=\cN_B[\chi_{AB}]$ is of the form $\ket{\chi}_{AB}=\sum_{x\in \ff_d}\sqrt{p_x}\ket{x}_A\ket{x}_B$. 

For our purposes, PDPC channels have two important properties. 
First, inputs diagonal in the standard basis lead to commuting outputs: $[\cN[\proj{x}],\cN[\proj{x'}]]=0$ for all $x,x'\in \ff_d$. 
Let us express this is an equivalent, but more useful form. 
Set $\ket{\sigma}_{ABE}=U_{BE}\ket{\chi}_{AB}\ket{0}_E$ with $\ket{\chi}_{AB}=\sum_{k\in \ff_d}\sqrt{p_k}\ket{k}_{A}\ket{k}_B$. 
Then the states $\rho_B(k)\coloneq \bra{k}_A\sigma_{AB}\ket{k}_A$ for different $k\in \ff_d$ commute. 

For the second property, we need a few more definitions. 
Set the unitary $U'_{AB}:\ket{x}_A\ket{y}_B\mapsto \ket{x}_A\ket{x-y}_B$ and, for $q_y=\sum_{k\in \ff_d} p_k |\alpha_{k,y}|^2\neq 0$, define pure states $\ket{\theta_y}=1/\sqrt{q_y}\sum_{k\in \ff_d}\sqrt{p_k}\alpha_{k,y} \ket{k}$ .
Then the states $\sigma_{AB}(x)\coloneq Z_A^x \sigma_{AB}Z_{A}^{-x}$ have the form 
\begin{equation}
\sigma_{AB}(x)= Z^x_AU_{AB}'^*\Big(\sum_{y\in \ff_d} q_y  \proj{\theta_y}_A \otimes \proj{y}_B\Big) U_{AB}'Z_A^{-x}\,.
\end{equation}
To see this, first note that $[U'_{AB},Z_A]=0$. 
Then, we have 
\begin{align}
U'_{AB}\sigma_{AB}U_{AB}'^*
&=\sum_{j,k,y\in \ff_d}\sqrt{p_j p_{k}} \alpha_{j,y}^*\alpha_{k,y} U'_{AB}(\ket{k}\bra{j}_A\otimes \ket{k-y}\bra{j-y}_B)U_{AB}'^*\\
&=\sum_{j,k,y\in \ff_d}\sqrt{p_j p_{k}} \alpha_{j,y}^*\alpha_{k,y} \ket{k}\bra{j}_A\otimes \ket{y}\bra{y}_B=\sum_{y\in \ff_d} q_y \proj{\theta_y}_A\otimes \proj{y}_B\,. 
\end{align}

The pure-state channel (PSC) in this decomposition have the form $x\mapsto\proj{\psi_x}$ with $\ket{\psi_x}=Z^x\ket\theta$ for $x\in \ff_d$. 
We call such channels symmetric $d$-ary PSCs. 
For this channel, the Gram matrix $M$ is formed using the pairwise output state overlaps,
\begin{equation}
M_{x,x'}=\braket{\psi_x|\psi_{x'}}=\bra\theta Z^{x'-x}\ket\theta\,.
\end{equation}
Since each entry depends only on $x'-x$ modulo $d$, $M$ is circulant and is diagonalized by the Fourier transform $F$. We call the vector $\blambda=(\lambda_0,\ldots,\lambda_{d-1})$ defined by
\begin{equation}\label{eq:eigen_list_definition}
F^\dagger M F=\operatorname{diag}(\lambda_0,\ldots,\lambda_{d-1})
\end{equation}
the \emph{eigen list} of the channel~\cite{mandal_belief_2026}. Its entries are indexed by Fourier vectors. Because $M$ is positive semidefinite and has unit diagonal, $\lambda_j\geq0$ and $\sum_j\lambda_j=d$. 
For a symmetric $d$-ary PSC, the eigen list $\blambda$ provides a compact description of the channel.

\subsection{Pure loss Bosonic channel}
\paragraph{Definition and properties}
The pure loss channel $\cN$ for a Bosonic mode $B$ and transmissivity $\eta\in [0,1]$ is given by the following unitary operation on $B$ and an environment mode $E$
\begin{equation}\label{eq:pure_loss_stinespring}
    U_{BE}:\ket{n}_{B}\otimes \ket{0}_{E}\mapsto \sum_{\ell=0}^{n}\alpha_{n,\ell}\ket{n-\ell}_{B}\otimes \ket{\ell}_E,
\end{equation}
where $|\alpha_{n,\ell}|^{2}=\binom{n}{\ell}(1-\eta)^{\ell}\eta^{n-\ell}$, $n$ is a nonnegative integer, and $\ket{k}$ for nonnegative integer $k$ is the $k$th Fock state.

When using coding schemes with the constraint that the encoder never outputs a state with energy per mode  larger than $\bar n$, the quantum capacity $C_Q (\mathcal{N},\bar n)$ of the pure loss channel is zero for $\eta\leq 1/2$ and, for $1/2<\eta\leq 1$ given by~\cite{wilde_quantum_2012a,wilde_energyconstrained_2018}
\begin{equation}
C_Q (\mathcal{N},\bar n) = g(\eta \bar n) - g((1-\eta)\bar n)\,,
\end{equation}
where $g(x)\coloneq (x+1)\log_2 (x+1)-x \log_2 x$. 
Furthermore, $C_Q(\mathcal N,\bar n) = -H(A|B)_{\sigma'}$ for $\sigma'_{AB} = \mathcal{I}_A \otimes \mathcal{N}_B [\chi_{AB}^{(\bar{n})}]$, where $\chi_{AB}^{(\bar{n})}$ is a purification of the thermal state of average energy $\bar{n} >0$, i.e. the two-mode squeezed state
\begin{equation}
\ket{\chi^{(\bar n)}}_{AB} = \frac{1}{\sqrt{\bar n+1}}\sum_{n=0}^\infty \left(\frac{\bar n}{\bar{n}+1}\right)^{n/2} \ket{n}_A \otimes \ket{n}_B\,.
\end{equation}
\paragraph{Approximation by truncation}
For $d$ prime, let $P_A$ be the projection onto the lowest $d>\bar n$ energy levels in mode $A$, i.e.\ up to Fock state $\ket{d-1}$, and define $\ket{\gamma^{(\bar n)}}_{AB}=P_A\ket{\chi^{(\bar n)}}_{AB}/\sqrt{\bra{\chi^{(\bar n)}}P\ket{\chi^{(\bar n)}}}$. 
The following shows that the error in the conditional entropy decays exponentially in the cutoff $d$:
\begin{lemma}\label{lem:truncation}
Applied to $\ket{\gamma^{(\bar n)}}_{AB}$, the pure loss channel is a PDPC. 
Call the output $\sigma_{AB}$. With $r=\tfrac{\bar n}{\bar n+1}$,
\begin{equation}
-H(A|B)_{\sigma}\geq -H(A|B)_{\sigma'}- 4r^{d/2} g(\bar n/r^{d/2})-2h_2(r^{d/2})\,. 
\end{equation}
\end{lemma}
\begin{proof}
Computing the squared overlap of the two input states we find  
\begin{align}
\braket{\chi^{(\bar n)}|\gamma^{(\bar n)}}^2
&={\bra{\chi^{(\bar n)}}P\ket{\chi^{(\bar n)}}}
=\frac{1}{\bar n+1}\sum_{n=0}^{d-1} r^n
=\frac{1}{\bar n+1}\frac{1-r^d}{1-r}
=1-r^d\,.
\end{align}
Therefore, by the relation between fidelity and trace distance, $\delta(\sigma'_{ABE},\sigma_{ABE})\leq r^{d/2}$. 
The proof is completed by appealing to Winter's entropy bound for energy-constrained states, in particular Lemma 15 of \cite{winter_tight_2016}. 

To do so, first note that both the $B$ and $E$ marginals are also close by monotonicity of the trace distance. 
For pure $\sigma$ we have $-H(A|B)_\sigma=H(B)_\sigma-H(E)_\sigma$, and similarly for $\sigma'$. 
Therefore the mean energy of $B$ is at most $\eta\bar n$ for $\sigma$ and $\sigma'$ (since truncating will just move weight to lower energies), while for $E$ the mean energy is at most $(1-\eta)\bar n$ in both states. 
Therefore Lemma 15 implies $|H(B)_\sigma-H(B)_{\sigma'}|\leq 2 r^{d/2}g(\bar nr^{-d/2})+h_2(r^{d/2})$. 
Combining this with the same argument for the entropy of the $E$ marginals gives the desired statement. 
\end{proof}

\section{Entanglement distillation}
\label{sec:entanglement_distillation}
In this section we give the details of the entanglement distillation procedure.  
\subsection{Setup}

Consider using a $(k,\ell,s_x,s_z,H)$ EACSS code $\cC$ of the form above for entanglement distillation from a state 
$\sigma_{AB}^{\otimes m}$ with $d=\text{dim}(A)$ prime. 
Suppose that for $H$ and some orthonormal basis  $\{\ket{z}_A\}_{z\in \ff_d}$ we have a set of POVMs indexed by $y^{s_z+\ell}\in \ff_d^{s_z+\ell}$ with elements $\Lambda_{B^m}(z^m;y^{s_z+\ell})$ such that 
\begin{equation}
\sum_{z^m\in \ff_d^m} \tr[(\proj{z^m}_{A^m}\otimes \Lambda_{B^m}(z^m,H_Zz^m))\sigma_{AB}^{\otimes m}]\geq 1-\eps_Z\,. 
\end{equation}
Furthermore, suppose that we also have a set of POVMs $\Gamma'_{A'^m B^m}(x^m;y^{s_x+\ell})$ such that 
\begin{equation}
\sum_{x^m\in \ff_d^m}\tfrac1{d^m}\tr[\Gamma'_{A'^mB^m}(x^m;G_X x^m) \, Z_{A'^m}^{-x^m}\sigma_{A'B}^{\otimes m} Z_{A'^m}^{x^m}]\geq 1-\eps_X\,. 
\end{equation}
Here $Z=\sum_{k\in \ff_d}\omega^{k}\proj{k}$, i.e.\ using the basis $\{\ket{k}_{A'}\}$. 
Define $f(x^m)$ to be the first $k$ entries in $Gx^m$. 
Then construct a POVM 
\begin{equation}
\Gamma_{A^m B^m}(z^k;y^{s_x+\ell})\coloneq\sum_{x^m:f(x^m)=z^k} \Gamma'_{A^m B^m}(x^m;y^{s_x+\ell}).
\end{equation}
Note that this POVM correctly determines $f(x^m)$ with probability at least $1-\eps_X$:
\begin{equation}
\sum_{x^m\in \ff_d^m}\tfrac1{d^m}\tr[\Gamma_{A^mB^m}(f(x^m);G_X x^m) \, Z_{A^m}^{-x^m}\sigma_{AB}^{\otimes m} Z_{A^m}^{x^m}]\geq 1-\eps_X\,. 
\end{equation}
Now define $U_\Lambda$ to be the following coherent implementation of $\Lambda$, with an output into an additional register $A'^m$: 
\begin{equation}
U_{\Lambda}:\ket{w^s}_{\widecheck{A}}\ket{y^\ell}_{\underline{A}}\ket{\psi}_{B^m}\ket{0}_{A'^m} \to \sum_{x^m\in \ff_d^m}\ket{w^{s_z}}_{\widecheck{A}}\ket{y^\ell}_{\underline{A}}\sqrt{\Lambda_{B^m}(x^m,w^{s_z}||y^\ell)}\ket{\psi}_{B^m}\ket{x^m}_{A'^m}\,. 
\end{equation}
Here $(a||b)$ denotes the concatenation of strings $a$ and $b$. 
Similarly, let $U_\Gamma$ to be a coherent implementation of $\Gamma$ on $A'^mB^m$, with an output into an additional register $\bar R$, and $W$ the final controlled-$\bar Z$ and $F^\dagger$ on each qudit in $\bar A$. 
Then define $\ket{\Psi}_{A^mB^mE^mA'^m\bar R}=WU_\Gamma U_\Lambda V\ket{\sigma}_{BE}^{\otimes m}\ket{0}_{A'^m}\ket{0}_{\bar R}$. This setup is depicted in Figure~\ref{fig:entanglementgeneration}. 

\subsection{Quality of the structured decoder}
The main result is the following bound on the fidelity of the decoder output with the maximally entangled state. 
\begin{theorem}
\label{thm:edcircuit}
For any $\sigma_{AB}$ and choice of orthonormal basis $\{\ket{k}_A\}$, 
the fidelity of $\Psi_{\bar A\bar R}$ with the maximally entangled state $\ket{\Phi}_{\bar A\bar R}$ satisfies 
\begin{equation}
F(\Phi_{\bar A \bar R},\Psi_{\bar A \bar R})\geq 1-(\sqrt{\eps_X}+\sqrt{\eps_Z})^2\,.
\end{equation}
\end{theorem}
\begin{proof}
Let $\ket{\sigma}_{ABE}$ be a purification of $\sigma_{AB}$ and define the pure states $\{\ket{\varphi_k}_{BE}\}_{k\in \ff_d}$ by $\sqrt{p_k}\ket{\varphi(k)}_{BE}={}_{A}\!\braket{k|\sigma}_{ABE}$. 
The state after the action of $V$ is 
\begin{equation}
\ket{\psi_1}\coloneq V\ket{\sigma}^{\otimes m}=\sum_{k^m} \sqrt{p_{k^m}} \ket{Hk^m}_{A^m} \ket{\varphi(k^m)}_{B^mE^m}\,.
\end{equation}
Applying $U_\Lambda$ with the ancilla $\ket{0}_{A'^m}$ gives 
\begin{equation}
\ket{\psi_2}\coloneq U_\Lambda\ket{\psi_1}=
\sum_{k^m,k'^m} \sqrt{p_{k^m}} \ket{Hk^m}_{A^m} \sqrt{\Lambda_{B^m}(k'^m,H_Zk^m)}\ket{\varphi(k^m)}_{B^mE^m}\ket{k'^m}_{A'^m}\,.
\end{equation}
By Lemma~\ref{lem:guessingfidelity}, $\braket{\psi_2'|\psi_2}\geq 1-\eps_Z$ for 
\begin{equation}
\ket{\psi_2'}\coloneq \sum_{k^m} \sqrt{p_{k^m}} \ket{Hk^m}_{A^m} \ket{\varphi(k^m)}_{B^mE^m}\ket{k^m}_{A'^m}\,.
\end{equation}

Continue with $\ket{\psi_2'}$. 
Expressing the $A^m$ system in the Fourier conjugate basis yields:
\begin{align}
\ket{\psi'_2}
&=\tfrac{1}{\sqrt{d^m}}\sum_{k^m,x^m} \sqrt{p_{k^m}}\ket{\widetilde{x^m}}_{A^m} \omega^{-x^m\cdot Hk^m}\ket{\varphi(k^m)}_{B^mE^m}\ket{k^m}_{A'^m}\\
&=\tfrac{1}{\sqrt{d^m}}\sum_{k^m,x^m} \sqrt{p_{k^m}}\ket{\widetilde{x^m}}_{A^m} \ket{\varphi(k^m)}_{B^mE^m}Z_{A'^m}^{-(x^m)^TH}\ket{k^m}_{A'^m}\\
&=\tfrac{1}{\sqrt{d^m}}\sum_{x^m} \ket{\widetilde{Gx^m}}_{A^m} Z_{A'^m}^{-x^m} \ket{\sigma}_{A'BE}^{\otimes m}\,.
\end{align} 
The Fourier transforms on the $\widehat{A}$ and $\underline{A}$ systems mean that subsequent control operations will extract the value of $G_Xx^m$. 
Then $U_\Gamma$ will yield
\begin{equation}
\ket{\psi_3'}\coloneq U_\Gamma\ket{\psi_2'}=\tfrac{1}{\sqrt{d^m}}\sum_{x^m,\bar{y}} \ket{\widetilde{Gx^m}}_{A^m} \sqrt{\Gamma_{A'^mB^m}(\bar{y},G_Xx^m)}\, Z_{A'^m}^{-x^m} \ket{\sigma}_{A'BE}^{\otimes m}\ket{\bar{y}}_{\bar R}\,. 
\end{equation}
Here we define $(\bar x,\check x, \hat x,\underline{x})=Gx^m$. 
By Lemma~\ref{lem:guessingfidelity}, $\braket{\psi_3''|\psi_3'}\geq 1-\eps_X$ for 
\begin{equation}
\ket{\psi_3''}\coloneq\tfrac{1}{\sqrt{d^m}}\sum_{x^m} \ket{\widetilde{Gx^m}}_{A^m} Z_{A'^m}^{-x^m} \ket{\sigma}_{A'BE}^{\otimes m}\ket{\bar{x}}_{\bar{R}}\,.
\end{equation}

Continuing with $\ket{\psi_3''}$, the controlled-$\bar Z$ from $\bar R$ to $A'^m$ removes the $\bar x$ dependence from the exponent; the most immediate way to accomplish this operation is to apply $V$ to $A'^m$ and then a controlled phase operation from $\bar R$ to $\bar A'$. 
Meanwhile, the inverse Fourier transform on $\bar A$ returns $\bar x$, meaning that the state $W\ket{\psi_3''}$ has the form $\ket{\Phi}_{\bar A\bar R}\otimes \ket{\xi}_{\widecheck{A}\widehat{A}\underline{A}B^mE^mA'^m}$ for a suitable $\ket{\xi}$. 

Now consider the fidelity of $W\ket{\psi''_3}$ with $WU_\Gamma U_\Lambda V\ket{\sigma}^{\otimes m}=WU_\Gamma \ket{\psi_2}$:
\begin{equation}
F(\ket{\Phi}_{\bar A\bar R}\ket{\xi},WU_\Gamma U_\Lambda V\ket{\sigma}^{\otimes m})
=|\bra{\psi_3''}U_\Gamma\ket{\psi_2}|.
\end{equation}
Applying the triangle inequality for fidelity (e.g.\ \cite[Proposition 10.5]{renes_quantum_2022}) gives 
\begin{align}
\arccos |\bra{\psi_3''}U_\Gamma\ket{\psi_2}| 
&\leq \arccos |\bra{\psi_3''}U_\Gamma\ket{\psi'_2}|+\arccos |\braket{\psi_2'|\psi_2}|\\
&\leq \arccos (1-\eps_X)+\arccos (1-\eps_Z)\,. 
\end{align}
Thus the fidelity of the $\bar A\bar R$ output with the maximally-entangled state is at least $1-(\sqrt{\eps_X}+\sqrt{\eps_Z})^2$. 
\end{proof}

\begin{lemma}\label{lem:resources}
The quantum operation $WU_\Gamma U_\Lambda V$ can be implemented by consuming $\ell$ entangled pairs and transmitting $s_x+s_z+2\ell$ dits of classical information. The net rate of entanglement distillation is therefore $\tfrac1m(k-\ell)\log_2 d$. 
\end{lemma}
\begin{proof}
Since the $\widecheck{A}$ and $\widehat{A}$ registers will anyway be discarded immediately after the single control operations, we may replace those control operations by measurement and classical communication. 
On the other hand, register $\underline{A}$ needs to be transmitted to Bob via teleportation, at a cost of $\ell$ entangled pairs. 
\end{proof}

\section{Polar Coding Scheme}\label{sec:polar}

In~\cite{renes_polar_2014}, quantum polar codes for qubit channels are obtained from CSS-type codes whose matrix is a polar transform. We follow this construction and choose the EACSS code of Theorem~\ref{thm:edcircuit} to be a polar code over $\ff_d$.
By Theorem~\ref{thm:edcircuit} and Lemma~\ref{lem:resources}, it suffices to specify three things. These are the matrix $H$, the partition of $[m]$ into the registers $\bar A,\widecheck A,\widehat A,\underline A$ and two decoders whose average block error probabilities vanish.
The amplitude information is the outcome $k^m$ of a measurement of Alice's registers $A^m$ in the standard basis. The phase information is the outcome $x^m$ of a measurement of $A^m$ in the Fourier conjugate basis. After $V$, the registers $A^m$ hold $Hk^m$ in the standard basis and $Gx^m$ in the Fourier basis.

\subsection{Polar transform}
Let $m$ be a power of $2$ and $\mathsf K_m=\mathsf B_m\mathsf K_2^{\otimes\log m}$ over $\ff_d$. Here $\mathsf K_2=\left(\begin{smallmatrix}1&0\\1&1\end{smallmatrix}\right)$ and $\mathsf B_m$ is the bit reversal permutation.
We take $H=\mathsf K_m^T$, so that $G=(H^T)^{-1}=\mathsf K_m^{-1}$.
The transformed amplitude information $u^m=Hk^m$ and the transformed phase information $v^m=Gx^m$ then satisfy
\begin{equation}\label{eq:polar_uv}
(u^m)^T=(k^m)^T\mathsf K_m\,,\qquad (x^m)^T=(v^m)^T\mathsf K_m^T\,.
\end{equation}
Let $J=\left(\begin{smallmatrix}0&1\\1&0\end{smallmatrix}\right)$. Then $J\mathsf K_2J=\mathsf K_2^T$, and $\mathsf B_m$ commutes with $J^{\otimes\log m}$ and $\mathsf K_2^{\otimes\log m}$. Hence $\mathsf K_m^T$ is $\mathsf K_m$ with the order of rows and columns reversed.
As in the qubit construction, the polar code for the phase information has the direction of all gates flipped. Its $i$th synthesized channel has the later values $v_{i+1}^m$ as side information, and the transformed phase information is decoded in the order $i=m,\dots,1$.
For $d=2$, we have $\mathsf K_m^{-1}=\mathsf K_m$ and the scheme reduces to the qubit construction.

\subsection{Amplitude and phase channels}
The amplitude CQ channel $W_A$ maps the amplitude information $k$ of one copy to the normalized state $W_A(k)=\rho_B(k)/p_k$ of $B$. This is the state $\varphi_B(k)$ in the proof of Theorem~\ref{thm:edcircuit}. Its input distribution is $\{p_k\}$ and its joint CQ state is $\rho_{AB}$ of Lemma~\ref{lem:coherentinfocalc}. By the first property of PDPC channels, its output states commute.
The phase CQ channel $W_P$ maps the phase information $x$ of one copy to the state $W_P(x)=\sigma_{A'B}(x)=Z_{A'}^x\sigma_{A'B}Z_{A'}^{-x}$ of Bob's systems $A'B$. Here $A'$ holds the decoded amplitude information and takes the role of $A$ in $\sigma_{AB}$. Its input distribution is uniform and its joint CQ state is $\tau_{XAB}$ of Lemma~\ref{lem:coherentinfocalc} with $A'$ in place of $A$.
By the second property of PDPC channels, the unitary $U'_{A'B}$ maps $W_P$ to the heralded mixture $x\mapsto\sum_yq_y\,Z_{A'}^x\proj{\theta_y}_{A'}Z_{A'}^{-x}\otimes\proj{y}_B$ of symmetric $d$-ary PSCs. The herald $y$ is stored in $B$ and has probability $q_y$. For the pure-loss channel, $y$ is the number of lost photons.

\subsection{Reliability parameter}
Let $W:u\mapsto\rho_u$ be a CQ channel with input $u\in\ff_d$, output system $Y$ and input distribution $\pi$. Consider the joint CQ state $\rho_{UY}=\sum_{u\in\ff_d}\pi(u)\proj{u}_U\otimes\rho_u$.
The pairwise fidelity of $W$ is
\begin{equation}\label{eq:fidelity_param}
F(W)=\frac1{d-1}\sum_{u\neq u'}\sqrt{\pi(u)\pi(u')}\,F(\rho_u,\rho_{u'})\,,\qquad F(\rho,\tau)=\big\|\sqrt\rho\sqrt\tau\big\|_1\,,
\end{equation}
where the sum is over ordered pairs.
It lies in $[0,1]$. It is $0$ if and only if the $\rho_u$ with $\pi(u)>0$ are perfectly distinguishable. It equals $1$ if and only if all $\rho_u$ coincide and $\pi$ is uniform.
The minimum discrimination error of $W$ is
\begin{equation}\label{eq:perr_def}
P_{\rm err}(W)=\min_{\{M_u\}_{u\in\ff_d}}\sum_{u\in\ff_d}\pi(u)\tr\big[(I-M_u)\rho_u\big]\,,
\end{equation}
where the minimum is over all POVMs $\{M_u\}_{u\in\ff_d}$ on $Y$.

Suppose that the output states $\rho_u=\sum_yP(y|u)\proj{y}$ commute. Then $W$ is a classical channel and $F(\rho_u,\rho_{u'})=\sum_y\sqrt{P(y|u)P(y|u')}$. In this case $F(W)$ equals the $d$-ary Bhattacharyya parameter
\begin{equation*}
Z(W)=\frac1{d-1}\sum_{u\neq u'}\sum_y\sqrt{P_{UY}(u,y)P_{UY}(u',y)}\,,\qquad P_{UY}(u,y)=\pi(u)P(y|u)\,.
\end{equation*}
In~\cite{mori_source_2014}, $Z(W)$ is used to analyze source and channel polarization over finite fields. For classical channels, $P_{\rm err}(W)\leq(d-1)Z(W)$~\cite[Lemma~22]{mori_source_2014}. The minimum in~\eqref{eq:perr_def} is attained by the maximum a posteriori decision. It equals $P_{\rm err}(W)=1-\sum_y\max_{u}P_{UY}(u,y)$.

Suppose instead that $W$ is a heralded mixture of symmetric $d$-ary PSCs with uniform $\pi$. That is, $W(u)=\sum_hq_h\,W_h(u)\otimes\proj{h}$ where each $W_h\colon u\mapsto\proj{\psi_u^h}$ is a symmetric $d$-ary PSC.
Then every $\rho_u$ is block diagonal in $h$. As a result, a measurement that first obtains the herald index $h$ is optimal and
\begin{equation}\label{eq:heralded_split}
P_{\rm err}(W)=\sum_hq_h\,P_{\rm err}(W_h)\,,\qquad F(W)=\sum_hq_h\,F(W_h)\,.
\end{equation}
The states $\{\ket{\psi_u^h}\}_{u\in\ff_d}$ of each $W_h$ are pure and geometrically uniform. For such states, the pretty good measurement attains $P_{\rm err}(W_h)$. 

In both cases,
\begin{equation}\label{eq:perr_bound}
P_{\rm err}(W)\leq(d-1)F(W)\,.
\end{equation}
For classical channels, this is the bound stated above. For heralded mixtures of symmetric PSCs, the error bound for the pretty good measurement in~\cite{barnum_reversing_2002} gives $P_{\rm err}(W_h)\leq(d-1)F(W_h)$ for every $h$. Averaging over $h$ with~\eqref{eq:heralded_split} gives~\eqref{eq:perr_bound}.

\subsection{Synthesized channels and speed of polarization}\label{sec:synth}
Let $K^m$ and $X^m$ be the amplitude and phase information of $A^m$, with $K^m$ distributed as $p_{k^m}=\prod_{t=1}^mp_{k_t}$ and $X^m$ uniform on $\ff_d^m$. Let $U^m=HK^m$ and $V^m=GX^m$.

The $i$th synthesized amplitude channel $W_{A,m}^{(i)}\colon u\mapsto\rho^{(i)}_{A,u}$ has input $U_i$ and output system $U^{i-1}B^m$. Its output state $\rho^{(i)}_{A,u}$ is the state of $U^{i-1}B^m$ conditioned on $U_i=u$, and its input distribution is $\pi_i(u)=\Pr[U_i=u]$.
Given $K^m=k^m$, the registers $U^m$ hold $Hk^m$ and $B^m$ is in the state $\varphi_B(k^m)=W_A^{\otimes m}(k^m)$. Summing over the amplitude information $k^m$ with $(Hk^m)_i=u$ gives
\begin{equation}\label{eq:synth_amp}
\rho^{(i)}_{A,u}=\frac1{\pi_i(u)}\sum_{k^m:\,(Hk^m)_i=u}p_{k^m}\,\proj{(Hk^m)^{i-1}}\otimes\varphi_B(k^m)\,,\qquad \pi_i(u)=\sum_{k^m:\,(Hk^m)_i=u}p_{k^m}\,,
\end{equation}
where $(Hk^m)^{i-1}$ is the vector of the first $i-1$ entries of $Hk^m$.
The state $W_A(k)$ is diagonal in the standard basis $\{\ket{b}\}_{b\in\ff_d}$ of $B$, so measuring $B$ in this basis loses no information about $K$. With this measurement, $K$ is a classical source with side information $b$ and joint distribution $P_{KB}(k,b)=p_kW_A(b|k)$, where $W_A(b|k)=|\alpha_{k,k-b}|^2$ with the subtraction modulo $d$. For the pure-loss channel, $b$ is the photon number at the output. The transform $U^m=HK^m$ is the source polarization transform of $K^m$ with kernel $\mathsf K_2$ over $\ff_d$~\cite{arikan_source_2010,mori_source_2014}. Each $W_{A,m}^{(i)}$ is then a classical channel, namely the $i$th synthesized source with side information $(u^{i-1},b^m)$, and $F(W_{A,m}^{(i)})=Z(W_{A,m}^{(i)})$ is its Bhattacharyya parameter. For the kernel $\mathsf K_2$, the Bhattacharyya parameters of the reliable synthesized sources decay as $2^{-m^\beta}$ for every $0<\beta<\tfrac12$. The following limit holds for every $0<\beta<\tfrac12$~\cite[Theorem~31 and Corollary~28]{mori_source_2014}:
\begin{equation}\label{eq:amp_polarization}
\lim_{m\to\infty}\frac1m\Big|\Big\{i\in[m]:Z\big(W_{A,m}^{(i)}\big)<2^{-m^\beta}\Big\}\Big|=1-\frac{H(A|B)_\rho}{\log d}\,.
\end{equation}
Here $H(A|B)_\rho$ is the conditional entropy of the amplitude information given $B$. It can be computed in the standard basis of $B$ because the outputs of $W_A$ commute.

The $i$th synthesized phase channel $W_{P,m}^{(i)}\colon v\mapsto\rho^{(i)}_{P,v}$ has input $V_i$ and output system $V_{i+1}^mA'^mB^m$. Its output state $\rho^{(i)}_{P,v}$ is the state of $V_{i+1}^mA'^mB^m$ conditioned on $V_i=v$, and its input distribution is uniform.
Given $X^m=x^m$, the registers $V^m$ hold $Gx^m$ and $A'^mB^m$ is in the state $W_P^{\otimes m}(x^m)=Z_{A'^m}^{x^m}\sigma_{A'B}^{\otimes m}Z_{A'^m}^{-x^m}$. Summing over the phase information $x^m$ with $(Gx^m)_i=v$ gives
\begin{equation}\label{eq:synth_phase}
\rho^{(i)}_{P,v}=\frac1{d^{m-1}}\sum_{x^m:\,(Gx^m)_i=v}\proj{(Gx^m)_{i+1}^m}\otimes Z_{A'^m}^{x^m}\sigma_{A'B}^{\otimes m}Z_{A'^m}^{-x^m}\,,
\end{equation}
where $(Gx^m)_{i+1}^m$ is the vector of the last $m-i$ entries of $Gx^m$.
After the order of the indices of both $x^m$ and $v^m$ is reversed, the transform $(x^m)^T=(v^m)^T\mathsf K_m^T$ becomes the standard transform $(x^m)^T=(v^m)^T\mathsf K_m$ of a polar channel code. Reversing the order of $x^m$ does not change $W_P^{\otimes m}$. The $W_{P,m}^{(i)}$ are the synthesized channels of the standard polar code with the index order reversed.
Polarization of CQ channels over prime fields with the parameter $F$ is proved in~\cite{nasser_polar_2018}. The following limit holds for every $0<\beta<\tfrac12$~\cite[Theorem~1]{nasser_polar_2018}:
\begin{equation}\label{eq:phase_polarization}
\lim_{m\to\infty}\frac1m\Big|\Big\{i\in[m]:F\big(W_{P,m}^{(i)}\big)<2^{-m^\beta}\Big\}\Big|=\frac{I(W_P)}{\log d}=1-\frac{H(X|AB)_\tau}{\log d}\,.
\end{equation}
Here $I(W_P)=\log d-H(X|AB)_\tau$ is the symmetric Holevo information of $W_P$.
The unitary $U'_{A'B}$ on Bob's systems maps $W_P$ to a heralded mixture of symmetric $d$-ary PSCs. In~\cite[Theorem~19]{mandal_quantum_2026a}, heralded mixtures of symmetric $d$-ary PSCs are shown to be closed under the check-node and bit-node channel combinations that build the synthesized channels of a polar code. It follows that each $W_{P,m}^{(i)}$ is unitarily equivalent to such a heralded mixture by a unitary on Bob's systems. Its herald consists of $v_{i+1}^m$, the heralds $y_t$ of the channel outputs and the heralds produced by the channel combinations.
Unitary equivalence preserves $F$ and $P_{\rm err}$, so~\eqref{eq:perr_bound} applies to both $W_{A,m}^{(i)}$ and $W_{P,m}^{(i)}$.

\subsection{Index sets and EACSS registers}
Fix $0<\beta<\tfrac12$, so that $m2^{-m^\beta}\to0$. The index $i$ is reliable for the amplitude information if the transformed amplitude information $u_i$ is nearly determined by $u^{i-1}$ and $B^m$. It is reliable for the phase information if the transformed phase information $v_i$ is nearly determined by $v_{i+1}^m$, $A'^m$ and $B^m$. The corresponding sets are defined by
\begin{equation}\label{eq:reliable_sets}
\mathcal G_A \coloneqq \big\{i\in[m]:Z\big(W_{A,m}^{(i)}\big)<2^{-m^\beta}\big\}\,,\qquad \mathcal G_P\coloneqq \big\{i\in[m]:F\big(W_{P,m}^{(i)}\big)<2^{-m^\beta}\big\}\,.
\end{equation}
They are the good sets for the amplitude and phase information. The corresponding bad sets are
\begin{equation}\label{eq:bad_sets}
\mathcal B_A\coloneqq[m]\setminus\mathcal G_A\,,\qquad \mathcal B_P\coloneqq[m]\setminus\mathcal G_P\,.
\end{equation}
By~\eqref{eq:perr_bound}, $P_{\rm err}(W_{A,m}^{(i)})\leq(d-1)2^{-m^\beta}$ for every $i\in\mathcal G_A$ and $P_{\rm err}(W_{P,m}^{(i)})\leq(d-1)2^{-m^\beta}$ for every $i\in\mathcal G_P$.
By~\eqref{eq:amp_polarization} and~\eqref{eq:phase_polarization},
\begin{equation}\label{eq:set_sizes}
\lim_{m\to\infty}\frac{|\mathcal B_A|}{m}=\frac{H(A|B)_\rho}{\log d}\,,\qquad \lim_{m\to\infty}\frac{|\mathcal B_P|}{m}=\frac{H(X|AB)_\tau}{\log d}\,.
\end{equation}

The two sets partition $[m]$ into the four registers of the EACSS code:
\begin{equation}\label{eq:registers}
\bar A=\mathcal G_A\cap\mathcal G_P\,,\qquad \widecheck A=\mathcal B_A\cap\mathcal G_P\,,\qquad \widehat A=\mathcal G_A\cap\mathcal B_P\,,\qquad \underline A=\mathcal B_A\cap\mathcal B_P\,.
\end{equation}
These are interleaved subsets of $[m]$. A permutation of the outputs of $V$ puts them in the order of the EACSS code without changing the code. The parameters of the code are $k=|\bar A|$, $s_z=|\widecheck A|$, $s_x=|\widehat A|$ and $\ell=|\underline A|$.
The frozen set of the transformed amplitude information is $\cF_Z=\mathcal B_A=\widecheck A\cup\underline A$. The rows of $H$ indexed by $\cF_Z$ form $H_Z$, so the $Z$-type syndrome $H_Zk^m=u_{\cF_Z}$ is the transformed amplitude information on the unreliable indices.
The frozen set of the transformed phase information is $\cF_X=\mathcal B_P=\widehat A\cup\underline A$. The rows of $G$ indexed by $\cF_X$ form $G_X$, so the $X$-type syndrome $G_Xx^m=v_{\cF_X}$ is the transformed phase information on the unreliable indices.
Since the four registers partition $[m]$,
\begin{equation}\label{eq:net_count}
|\bar A|-|\underline A|=m-|\cF_Z|-|\cF_X|\,.
\end{equation}

\subsection{Decoding the amplitude information}
The frozen values $u_{\cF_Z}=H_Zk^m$ of the transformed amplitude information are not fixed in advance. They are random with the distribution induced by $p_{k^m}$, and they form the syndrome of a polar source code for the amplitude information $K^m$ with side information $B^m$.
By~\eqref{eq:set_sizes}, the syndrome consists of $|\cF_Z|$ symbols with $|\cF_Z|/m\to H(A|B)_\rho/\log d$. This is the rate of source coding with quantum side information~\cite{devetak_classical_2003}.
The nonuniform distribution of $K$ is supplied by the state $\ket{\sigma}$, so the source code needs no input shaping.

A decoder for this source code is a family of POVMs $\{\Lambda^{(y)}_{w}\}_{w}$ on $B^m$, one for each frozen value $y\in\ff_d^{|\cF_Z|}$, with outcomes $w\in\ff_d^{|\mathcal G_A|}$ that estimate the transformed amplitude information on $\mathcal G_A$. Its average error probability is
\begin{equation}\label{eq:PeZ_def}
\bar P_{e,Z}=\sum_{k^m}p_{k^m}\tr\big[\big(I-\Lambda^{(H_Zk^m)}_{(Hk^m)_{\mathcal G_A}}\big)\varphi_B(k^m)\big]\,.
\end{equation}
The decoder defines the POVMs of Theorem~\ref{thm:edcircuit} by
\begin{equation}\label{eq:Lambda_from_polar}
\Lambda_{B^m}(k^m;y)=\begin{cases}\Lambda^{(y)}_{(Hk^m)_{\mathcal G_A}} & H_Zk^m=y\,,\\ \mathbf 0&\text{otherwise.}\end{cases}
\end{equation}
Since $H$ is invertible, $k^m\mapsto(Hk^m)_{\mathcal G_A}$ is a bijection from $\{k^m:H_Zk^m=y\}$ onto $\ff_d^{|\mathcal G_A|}$. It follows that $\{\Lambda_{B^m}(k^m;y)\}_{k^m}$ is a POVM for every $y$ and
\begin{equation}\label{eq:epsZ}
\sum_{k^m}p_{k^m}\tr\big[\Lambda_{B^m}(k^m;H_Zk^m)\varphi_B(k^m)\big]=1-\bar P_{e,Z}\,.
\end{equation}
This is the hypothesis of Theorem~\ref{thm:edcircuit} on the POVMs $\Lambda_{B^m}(k^m;y)$ with $\eps_Z=\bar P_{e,Z}$.


The successive-cancellation decoder decides $u_1,\dots,u_m$ in this order. For $i\in\cF_Z$ it sets $\hat u_i=u_i$. For $i\in\mathcal G_A$ it takes the maximum a posteriori estimate of $u_i$ given $\hat u^{i-1}$ and the outcomes $b^m$ of a measurement of $B^m$ in the standard basis. The posterior follows from $u^m=Hk^m$ and the probability that the amplitude information is $k^m$ and this measurement gives $b^m$,
\begin{equation}\label{eq:amp_joint}
p_{k^m}\bra{b^m}\varphi_B(k^m)\ket{b^m}=\prod_{t=1}^mp_{k_t}W_A(b_t|k_t)\,,
\end{equation}
with $W_A(b|k)$ as defined in Section~\ref{sec:synth}. The successive-cancellation decoder starts from the probabilities $p_kW_A(b_t|k)$ for $k\in\ff_d$ and $t\in[m]$. In this case, the decoder uses the prior $p_k$ due to the fact that $K$ is not uniform.
The resulting POVM elements are projectors onto sets of strings $b^m$. They are diagonal in the standard basis, and $U_\Lambda$ is the reversible computation of the estimate $\hat k^m$ from $b^m$ and $u_{\cF_Z}$.
The decoder error only happens if a decision on some $i\in\mathcal G_A$ is wrong while all earlier decisions are correct. The union bound over $i\in\mathcal G_A$ and~\eqref{eq:perr_bound} give
\begin{equation}\label{eq:PeZ_bound}
\bar P_{e,Z}\leq\sum_{i\in\mathcal G_A}P_{\rm err}\big(W_{A,m}^{(i)}\big)\leq(d-1)\sum_{i\in\mathcal G_A}Z\big(W_{A,m}^{(i)}\big)\leq(d-1)\,m\,2^{-m^\beta}\,.
\end{equation}

\subsection{Decoding the phase information}
The frozen values $v_{\cF_X}=G_Xx^m$ of the transformed phase information are uniformly distributed because $X^m$ is uniform and $G$ is invertible. They form the frozen part of a polar channel code for $W_P$ with uniform input, where Bob holds $A'^m$ and $B^m$ as channel output.
By~\eqref{eq:set_sizes}, there are $|\cF_X|$ frozen values with $|\cF_X|/m\to H(X|AB)_\tau/\log d$.

A decoder for this channel code is a family of POVMs $\{\Gamma^{(y)}_{w}\}_{w}$ on $A'^mB^m$, one for each frozen value $y\in\ff_d^{|\cF_X|}$, with outcomes $w\in\ff_d^{|\mathcal G_P|}$ that estimate the transformed phase information on $\mathcal G_P$. Its average error probability is
\begin{equation}\label{eq:PeX_def}
\bar P_{e,X}=\sum_{x^m}\frac1{d^m}\tr\big[\big(I-\Gamma^{(G_Xx^m)}_{(Gx^m)_{\mathcal G_P}}\big)W_P^{\otimes m}(x^m)\big]\,.
\end{equation}
The decoder defines the POVMs of Theorem~\ref{thm:edcircuit} by
\begin{equation}\label{eq:Gamma_from_polar}
\Gamma'_{A'^mB^m}(x^m;y)=\begin{cases}\Gamma^{(y)}_{(Gx^m)_{\mathcal G_P}} & G_Xx^m=y\,,\\ \mathbf 0&\text{otherwise.}\end{cases}
\end{equation}
Since $G$ is invertible, $x^m\mapsto(Gx^m)_{\mathcal G_P}$ is a bijection from $\{x^m:G_Xx^m=y\}$ onto $\ff_d^{|\mathcal G_P|}$. It follows that $\{\Gamma'_{A'^mB^m}(x^m;y)\}_{x^m}$ is a POVM for every $y$ and
\begin{equation}\label{eq:epsX}
\sum_{x^m}\frac1{d^m}\tr\big[\Gamma'_{A'^mB^m}(x^m;G_Xx^m)\,W_P^{\otimes m}(x^m)\big]=1-\bar P_{e,X}\,.
\end{equation}
This is the success probability of $\Gamma'$ on the states $W_P^{\otimes m}(x^m)$. Theorem~\ref{thm:edcircuit} requires it for the states $Z_{A'^m}^{-x^m}\sigma_{A'B}^{\otimes m}Z_{A'^m}^{x^m}=W_P^{\otimes m}(-x^m)$. Since $x^m\mapsto-x^m$ is a bijection of $\ff_d^m$ and $G_X(-x^m)=-G_Xx^m$, the POVMs $\Gamma'_{A'^mB^m}(-x^m;-y)$ satisfy
\begin{equation*}
\sum_{x^m}\frac1{d^m}\tr\big[\Gamma'_{A'^mB^m}(-x^m;-G_Xx^m)\,W_P^{\otimes m}(-x^m)\big]=1-\bar P_{e,X}\,,
\end{equation*}
which is the hypothesis of Theorem~\ref{thm:edcircuit} with $\eps_X=\bar P_{e,X}$. Bob implements them by negating the received frozen value before decoding and negating the decoder outcome. The POVM $\Gamma$ in the setup of Theorem~\ref{thm:edcircuit} is obtained by summing over the $x^m$ with the same value $f(x^m)=(Gx^m)_{\bar A}$. It determines $f(x^m)$ with probability at least $1-\bar P_{e,X}$.

Unlike the amplitude information, the phase information can be decoded with a single measurement for all frozen values. The channel satisfies $W_P(x)=Z_{A'}^xW_P(0)Z_{A'}^{-x}$. For a frozen value $y$, let $c^m\in\ff_d^m$ be the vector with $(Gc^m)_{\cF_X}=y$ and $(Gc^m)_{\mathcal G_P}=\mathbf 0$. The map $x^m\mapsto x^m-c^m$ sends $\{x^m:G_Xx^m=y\}$ onto $\{x^m:G_Xx^m=\mathbf 0\}$ and leaves $(Gx^m)_{\mathcal G_P}$ unchanged. A decoder $\{\Gamma^{(\mathbf 0)}_w\}_w$ for the all-zero frozen value then gives the decoder
\begin{equation}\label{eq:phase_covariance}
\Gamma^{(y)}_w=Z_{A'^m}^{c^m}\,\Gamma^{(\mathbf 0)}_w\,Z_{A'^m}^{-c^m}
\end{equation}
for the frozen value $y$, with the same error probability for every $y$.
Since $U_{A'B}'^{\otimes m}$ is a unitary on Bob's systems, a decoder $\{\tilde\Gamma_w\}_w$ for the heralded mixture obtained by $U_{A'B}'^{\otimes m}$ gives the decoder $\{U_{A'B}'^{\dagger\,\otimes m}\tilde\Gamma_w U_{A'B}'^{\otimes m}\}_w$ for $W_P$ with the same error probability.

The successive-cancellation decoder decides $v_m,\dots,v_1$ in this order. For $i\in\cF_X$ it sets $\hat v_i=v_i$. For $i\in\mathcal G_P$ it performs a projective measurement $\{\Pi^{(i)}_{v|\hat v_{i+1}^m}\}_{v\in\ff_d}$ on $A'^mB^m$ that attains $P_{\rm err}(W_{P,m}^{(i)})$ given $\hat v_{i+1}^m$. By~\eqref{eq:heralded_split}, this measurement reads the herald of $W_{P,m}^{(i)}$ and applies the pretty good measurement to the corresponding PSC.
The POVM element of the decoder for the outcome $v_{\mathcal G_P}$ is
\begin{equation}\label{eq:phase_sc_povm}
\Gamma^{(v_{\cF_X})}_{v_{\mathcal G_P}}=\Pi^{(m)}_{v_m|v_{m+1}^m}\cdots\Pi^{(1)}_{v_1|v_2^m}\cdots\Pi^{(m)}_{v_m|v_{m+1}^m}\,,
\end{equation}
where $\Pi^{(i)}=I$ for $i\in\cF_X$.
The output states of $W_{P,m}^{(i)}$ do not commute in general, and the projectors for different $i$ do not commute. The union bound for sequential projective measurements~\cite{gao_quantum_2015} and~\eqref{eq:perr_bound} give
\begin{equation}\label{eq:PeX_bound}
\bar P_{e,X}\leq4\sum_{i\in\mathcal G_P}P_{\rm err}\big(W_{P,m}^{(i)}\big)\leq4(d-1)\sum_{i\in\mathcal G_P}F\big(W_{P,m}^{(i)}\big)\leq4(d-1)\,m\,2^{-m^\beta}\,.
\end{equation}

\subsection{Entanglement Generation Rate}
By~\eqref{eq:net_count} and~\eqref{eq:set_sizes}, the net rate of the polar code satisfies
\begin{equation}\label{eq:polar_rate}
\lim_{m\to\infty}\frac{|\bar A|-|\underline A|}{m}\log d=\log d-H(A|B)_\rho-H(X|AB)_\tau=-H(A|B)_\sigma\,,
\end{equation}
where the second equality is from Lemma~\ref{lem:coherentinfocalc}. 

\begin{theorem}\label{thm:polar_scheme}
Let $0<\beta<\tfrac12$ and let the registers be given by~\eqref{eq:registers}. Let $\{\Lambda^{(y)}_w\}$ and $\{\Gamma^{(y)}_w\}$ be decoders for the transformed amplitude and phase information with average error probabilities $\bar P_{e,Z}$ and $\bar P_{e,X}$.
Then the scheme of Theorem~\ref{thm:edcircuit} with $H=\mathsf K_m^T$ generates $|\bar A|$ maximally entangled pairs of local dimension $d$. It consumes $|\underline A|$ such pairs and transmits $|\widecheck A|+|\widehat A|+2|\underline A|$ dits of classical information. The fidelity of its output with the maximally entangled state satisfies the bound of Theorem~\ref{thm:edcircuit} with $\eps_Z=\bar P_{e,Z}$ and $\eps_X=\bar P_{e,X}$.
For the successive-cancellation decoders, $\bar P_{e,Z}\leq(d-1)m2^{-m^\beta}$ and $\bar P_{e,X}\leq4(d-1)m2^{-m^\beta}$. The net rate $(|\bar A|-|\underline A|)\log d/m$ converges to $-H(A|B)_\sigma$.
\end{theorem}

\begin{proof}
The matrix $\mathsf K_2^{\otimes\log m}$ is lower triangular with unit diagonal, so $H=\mathsf K_m^T$ is invertible over $\ff_d$. Up to a permutation of the outputs of $V$, it defines an EACSS code with the registers~\eqref{eq:registers}.
By~\eqref{eq:epsZ} and~\eqref{eq:epsX}, the decoders define POVMs $\Lambda$ and $\Gamma'$ that satisfy the hypotheses of Theorem~\ref{thm:edcircuit} with $\eps_Z=\bar P_{e,Z}$ and $\eps_X=\bar P_{e,X}$. Their coherent implementations are the unitaries $U_\Lambda$ and $U_\Gamma$ in the setup of Theorem~\ref{thm:edcircuit}, which gives the fidelity bound.
The resources follow from Lemma~\ref{lem:resources}. The error bounds are~\eqref{eq:PeZ_bound} and~\eqref{eq:PeX_bound}, and the rate is~\eqref{eq:polar_rate}.
\end{proof}
We now show that the scheme achieves the energy-constrained capacity of the pure-loss channel. For $\eta\leq\tfrac12$ the capacity is zero, so let $\eta>\tfrac12$ and $R<C_Q(\cN,\bar n)$. The correction terms in Lemma~\ref{lem:truncation} vanish as $d\to\infty$. Hence there is a prime $d>\bar n$ for which the truncated input $\ket{\gamma^{(\bar n)}}$ gives $-H(A|B)_\sigma>R$. The mean photon number of each channel input is $\bar n-dr^d/(1-r^d)\leq\bar n$, so the energy constraint holds. For this fixed $d$, Theorem~\ref{thm:polar_scheme} with $0<\beta<\tfrac12$ gives codes whose net rate converges to $-H(A|B)_\sigma>R$ as $m\to\infty$. Their error probabilities are at most $4(d-1)m2^{-m^\beta}\to0$, so the fidelity tends to $1$. Thus every rate below $C_Q(\cN,\bar n)$ is achievable. The limit $m\to\infty$ is taken at fixed $d$.

More generally, let $\cN$ be a degradable PDPC channel. Its quantum capacity is the maximum of $-H(A|B)_\sigma$ over input states. By phase covariance and the concavity of the coherent information for degradable channels, the maximum is attained by an input $\ket\chi=\sum_k\sqrt{p_k}\ket k\ket k$. Theorem~\ref{thm:polar_scheme} with this $p_k$ then achieves the quantum capacity.

\section{Code Construction and Efficient Implementation}\label{sec:bpqm}
In the previous section, we showed that CQ polar decoding of the amplitude and phase information gives a reliable entanglement-generation protocol.
In this section, we describe how to construct the partition $\bar A,\widecheck A,\widehat A,\underline A$ for a $d$-ary polar code of blocklength $m$ on a PDPC channel. We also show that the resulting protocol can be realized efficiently.
We describe the decoder for the amplitude information, which is a classical successive-cancellation computation that Bob runs reversibly.
Next we describe the decoder for the phase information. Its decisions are made by $d$-ary BPQM on heralded mixtures of symmetric pure-state channels.
We describe how the good and bad sets of the amplitude and phase information are chosen for target error probabilities using density evolution.
Finally, we show that every step of the protocol uses a number of operations polynomial in $m$ and $d$.

\subsection{Amplitude decoder}\label{sec:amp_decoder}
The decoder for the amplitude information is the successive-cancellation decoder of~\cite{arikan_source_2010} over $\ff_d$.
It starts from the probabilities $p_kW_A(b_t|k)$ of~\eqref{eq:amp_joint} and computes the a posteriori probabilities of $u_i$ given $\hat u^{i-1}$ and $b^m$ through the $\log m$ levels of recursions of the kernel $\mathsf K_2$.
Since $(u_1,u_2)=(k_1+k_2,k_2)$ for $\mathsf K_2$, each step maps two probability vectors in $\mathbb R^d$ to one. It uses a cyclic convolution for the first input of $\mathsf K_2$ and a pointwise product for the second input given the decision on the first.
Each step uses $O(d^2)$ arithmetic operations and the decoder performs $O(m\log m)$ steps. In total the decoder uses $O(d^2m\log m)$ arithmetic operations, which reduces to the $O(m\log m)$ of~\cite{arikan_source_2010} for $d=2$.
The estimate $\hat k^m$ follows from $(\hat k^m)^T=(\hat u^m)^T\mathsf K_m^{-1}$ with $O(m\log m)$ operations over $\ff_d$.
Bob runs these computations as a reversible circuit and then uncomputes all intermediate values. This gives $U_\Lambda$ with at most twice the number of operations. Since the POVM elements of $\Lambda$ are projectors in the standard basis of $B^m$, $U_\Lambda$ is a coherent implementation as in Lemma~\ref{lem:guessingfidelity}.

\subsection{Phase decoder}\label{sec:phase_decoder}
In the heralded mixture of Section~\ref{sec:polar}, consider the PSC $x\mapsto Z^x\proj{\theta_y}Z^{-x}$ with herald $y$ and $q_y>0$. Its Gram matrix has entries $M^{(y)}_{x,x'}=\sum_k p_k|\alpha_{k,y}|^2\omega^{(x'-x)k}/q_y$. Applying the Fourier convention in~\eqref{eq:eigen_list_definition} gives its eigen list
\begin{equation}\label{eq:eigen_list_y}
\blambda_y=\Big[\frac{d\,p_{-j}\,|\alpha_{-j,y}|^2}{q_y}\Big]_{j\in\ff_d}\,,
\end{equation}
where the indices are taken modulo $d$. The heralded mixture is of the form in~\cite[Definition~17]{mandal_quantum_2026a}, so the results of~\cite{mandal_quantum_2026a} used below apply to it.
With the index order reversed as in Section~\ref{sec:polar}, the synthesized phase channels follow the recursions of polar codes by check node and bit node channel combining in~\cite{mandal_belief_2026}.

For $i\in\mathcal G_P$, the BPQM unitary $U_i$ on $A'^mB^m$ is composed of the check node unitary $U^{\cnop}$ and the bit node unitaries $U^{\vnop}_{\blambda_1,\blambda_2}$ of~\cite{mandal_belief_2026} as extended in~\cite[Remark~14]{mandal_belief_2026}.
These unitaries are controlled on the heralds as in~\cite[Remark~15]{mandal_belief_2026}. The values $v_{i+1}^m$ enter through the shifts described in Section~\ref{sec:implementation}.
The eigen lists $\blambda_1$ and $\blambda_2$ of each bit node unitary follow from~\eqref{eq:eigen_list_y} and the heralds by~\cite[Lemmas~12 and~13]{mandal_belief_2026}.
By~\cite[Theorem~19]{mandal_quantum_2026a}, the pretty good measurement $\{\Pi_v\}_{v\in\ff_d}$ of~\cite[Lemma~18]{mandal_quantum_2026a} on the final qudit of $U_i$ is optimal. The projectors
\begin{equation}\label{eq:bpqm_projector}
\Pi^{(i)}_{v|v_{i+1}^m}=U_i^\dagger\big(\Pi_v\otimes I\big)U_i
\end{equation}
then attain $P_{\rm err}(W_{P,m}^{(i)})$ as required in Section~\ref{sec:polar}.
For each $i\in\mathcal G_P$, Bob adds the outcome coherently into a decision qudit in the state $\ket0$. For $i\in\bar A$ he also subtracts the decision qudit from the corresponding qudit of $\bar R$, which then holds the negated decision required in Section~\ref{sec:polar}. He then applies $U_i^\dagger$ as in~\cite{mandal_quantum_2026a}. These steps for $i=m,\dots,1$ form the successive-cancellation decoder of Section~\ref{sec:polar}. Its coherent implementation is given in Section~\ref{sec:implementation}.

\subsection{Code construction using Density evolution}\label{sec:construction}
We use density evolution to design practical polar codes at a chosen blocklength and target error probability. The polarization and rate guarantees are established in Section~\ref{sec:polar}. Here we estimate the synthesized-channel errors numerically and use them to select the amplitude and phase information sets.
Density evolution (DE) is a standard method to analyze message-passing decoders on tree factor graphs~\cite{richardson2008modern}. It tracks the distribution of the messages through the node updates and estimates the error probability of each decoded symbol. For classical polar codes, the synthesized channels can be estimated by Monte Carlo simulation of the successive-cancellation decoder~\cite{arikan_channel_2009}. For BPQM on CQ channels, Monte Carlo DE was introduced for binary LDPC codes in~\cite{brandsen2022belief} and extended to binary polar codes in~\cite{mandal_polar_2024}. For symmetric $d$-ary PSCs, the messages of BPQM are heralded mixtures described by eigen lists, and DE tracks a population of eigen lists through the check node and bit node updates~\cite{mandal_belief_2026}. We use this eigen-list DE for the phase information and the Monte Carlo construction for the amplitude information.

For the phase information, a bag is a list of $L$ eigen lists as in~\cite{mandal_belief_2026}. It represents a heralded mixture whose herald is uniform over the list. The initial bag contains the eigen lists $\blambda_y$ of~\eqref{eq:eigen_list_y} with $y$ drawn independently with probability $q_y$.
A check node update pairs the entries of two bags in random order. For each pair it draws the herald $h$ with its probability in~\cite[Lemma~12]{mandal_belief_2026} and keeps the eigen list $\blambda^{(\cnop,h)}$ of the check-node output for this herald. A bit node update keeps the eigen list $\blambda^{\vnop}$ of the bit-node output~\cite[Lemma~13]{mandal_belief_2026} for each pair.
Applying these updates along the $\log m$ levels of the recursions in the reversed index order of Section~\ref{sec:polar} gives a bag for every synthesized phase channel~\cite[Algorithm~3]{mandal_belief_2026}. The estimate $\hat P_{\rm err}(W_{P,m}^{(i)})$ is the average of $1-\big(\frac1d\sum_j\sqrt{\lambda_j}\big)^2$ over the eigen lists of its bag~\cite[Lemma~9]{mandal_belief_2026}. By~\eqref{eq:heralded_split}, this average estimates the minimum discrimination error of the heralded mixture.

For the amplitude information, the successive-cancellation decoder of Section~\ref{sec:amp_decoder} is simulated $L$ times~\cite{arikan_channel_2009}. Each run draws $k^m$ with probability $p_{k^m}$ and $b^m$ with probability $\prod_tW_A(b_t|k_t)$. At index $i$ the decoder is given the correct values $u^{i-1}$. The fraction of runs in which the maximum a posteriori estimate of $u_i$ is wrong is the estimate $\hat P_{\rm err}(W_{A,m}^{(i)})$.

Fix target error probabilities $\eps_Z$ and $\eps_X$. The set $\mathcal G_A$ is formed by adding indices in increasing order of $\hat P_{\rm err}(W_{A,m}^{(i)})$ as long as $\sum_{i\in\mathcal G_A}\hat P_{\rm err}(W_{A,m}^{(i)})\leq\eps_Z$. The set $\mathcal G_P$ is formed in the same way with the condition $4\sum_{i\in\mathcal G_P}\hat P_{\rm err}(W_{P,m}^{(i)})\leq\eps_X$, as in~\cite[Lemma~17]{mandal_belief_2026}. The factor $4$ is that of~\eqref{eq:PeX_bound}.
The selected sets give the registers $\bar A,\widecheck A,\widehat A,\underline A$ by~\eqref{eq:registers} and the design rate $(|\bar A|-|\underline A|)\log d/m$. The estimated error sums are used as design criteria motivated by~\eqref{eq:PeZ_bound} and~\eqref{eq:PeX_bound}. Substituting these estimates into the fidelity expression of Theorem~\ref{thm:polar_scheme} gives a numerical prediction of the achieved fidelity bound; the theorem itself applies to the true decoder error probabilities. The population size $L$ controls the numerical resolution, and the stability of the selected sets can be checked by increasing $L$ and repeating the construction.

\subsection{Decoder implementation and complexity}\label{sec:implementation}
The amplitude decoder of Section~\ref{sec:amp_decoder} is the successive-cancellation decoder for source polarization over $\ff_d$~\cite{arikan_source_2010,mori_source_2014}. Its reversible implementation $U_\Lambda$ uses $O(d^2m\log m)$ arithmetic operations.

For the phase information, Bob runs the BPQM unitaries coherently. In binary BPQM, the parameters that describe each message can be kept in ancilla registers, so that the node unitaries are controlled on them coherently~\cite{piveteau_quantum_2022}. We use the same idea for $d$-ary messages and work in the real-register model of~\cite{piveteau_efficient_2025}. In this model a real register stores a real number exactly, and single-qudit gates, two-qudit gates, gates controlled by real registers and arithmetic operations on real registers each have unit cost. Every message carries an eigen-list register, which consists of $d$ real registers, one for each entry of the eigen list of the message. Its state $\ket{\mathbf 0}$ has all entries equal to $0$. Since the entries of an eigen list sum to $d$, the state $\ket{\mathbf 0}$ differs from $\ket{\blambda}$ for every eigen list $\blambda$.

During the check node operation, the check node unitary $U^{\cnop}$ maps two messages to a heralded mixture as in Section~\ref{sec:polar}. Its herald $h\in\ff_d$ is held in the standard basis of the second qudit~\cite[Lemma~12]{mandal_belief_2026}. We write $\blambda^{(\cnop,h)}$ for the eigen list of this heralded mixture for the herald value $h$, as given in~\cite[Lemma~12]{mandal_belief_2026}.
Let $\mathbf h$ be the string of values of all heralds below a message, including the heralds $y$ of the channel outputs. The eigen list of the message is a function $\blambda(\mathbf h)$ of this string. Bob writes it into an eigen-list register by the isometry $\ket{\mathbf h}\mapsto\ket{\mathbf h}\ket{\blambda(\mathbf h)}$. He realizes this isometry by the unitary
\begin{equation}\label{eq:herald_to_register}
\sum_{\mathbf h}\proj{\mathbf h}\otimes\Big(\ketbra{\mathbf 0}{\blambda(\mathbf h)}+\ketbra{\blambda(\mathbf h)}{\mathbf 0}+I-\proj{\mathbf 0}-\proj{\blambda(\mathbf h)}\Big)
\end{equation}
on the heralds and a register in $\ket{\mathbf 0}$. This unitary is well defined because the heralds are held in the standard basis. It leaves the heralds in place for the inverse node unitaries.

Bob does not apply~\eqref{eq:herald_to_register} to all heralds at once. He builds it from local updates that follow the node unitaries. At the start, he computes $\blambda_y$ of~\eqref{eq:eigen_list_y} from the herald $y$ in the system $B$ of each copy into an eigen-list register for this copy.
A check node combines two messages whose eigen-list registers hold $\blambda_1$ and $\blambda_2$. Bob first applies $U^{\cnop}$ to their two qudits of $A'^m$. The second qudit then holds the new herald $h$, and Bob applies
\begin{equation*}
\ket{\blambda_1}\ket{\blambda_2}\ket h\ket{\mathbf 0}\mapsto\ket{\blambda_1}\ket{\blambda_2}\ket h\ket{\blambda^{(\cnop,h)}}
\end{equation*}
to the two eigen-list registers, the herald qudit and a new eigen-list register.
At a bit node, Bob applies the bit node unitary controlled on the two eigen-list registers
\begin{equation}\label{eq:register_bitnode}
\sum_{\blambda_1,\blambda_2}\proj{\blambda_1}\otimes\proj{\blambda_2}\otimes U^{\vnop}_{\blambda_1,\blambda_2}
\end{equation}
to the two qudits. Here $U^{\vnop}_{\blambda_1,\blambda_2}$ is the bit node unitary that~\cite{mandal_belief_2026} constructs from the eigen lists $\blambda_1$ and $\blambda_2$, and the sum runs over the eigen lists that the two registers can hold. He then computes the eigen list $\blambda^{\vnop}$ of~\cite[Lemma~13]{mandal_belief_2026} from the two eigen-list registers into a new eigen-list register.
The new register belongs to the combined message and is read by the next node. The two incoming registers are kept for the inverse node unitary. To invert a node, Bob first uncomputes the new register and then applies the inverse node unitary.

By induction over the nodes, the eigen-list register of every message holds $\blambda(\mathbf h)$ for every string $\mathbf h$ of herald values below it. The gate~\eqref{eq:register_bitnode} then applies the same unitary as the herald-controlled gate of~\cite[Remark~15]{mandal_belief_2026} for every herald value. The pretty good measurement does not depend on the eigen lists and needs no control. The decoder then realizes the projectors~\eqref{eq:bpqm_projector}. These updates replace the classical update of the conditional parameter lists in~\cite{mandal_polar_2024}.

Bob decodes the phase information by the polar decoder of~\cite{mandal_polar_2024}, using the check-node and bit-node unitaries of Section~\ref{sec:phase_decoder}. The decoding factor graph of an index $i$ is the factor graph of the synthesized channel $W_{P,m}^{(i)}$, and its root qudit carries the message for $v_i$. To decode $v_m$, Bob constructs its decoding factor graph and applies the check-node unitaries of this graph.
At the root of an index in $\mathcal G_P$, Bob adds the outcome of the pretty good measurement coherently into its decision qudit and applies $\sum_{v\in\ff_d}\proj v\otimes Z^{-v}$ to the decision qudit and the root qudit. For an index in $\bar A$, he also subtracts the decision qudit from the corresponding qudit of $\bar R$, since the decoder outcome is negated as in Section~\ref{sec:polar}. The decoder also uses the negated frozen values. At a frozen index Bob therefore applies $Z^{v}$ with the transmitted value $v$ of $\widehat A$, or $\sum_{v\in\ff_d}\proj v\otimes Z^{v}$ to the corresponding qudit of $\underline A$ and the root qudit. After the inverse Fourier transform in Figure~\ref{fig:entanglementgeneration}, this qudit holds $v$ in the standard basis. These shifts remove the decided or frozen value from the root qudit, so that the later steps act as for the value $0$. They do not change the eigen lists. They generalize the bit flip applied after each decision in the binary decoder of~\cite{mandal_polar_2024}.
To decode each later value, Bob first applies the inverses of the node unitaries of the previous decoding factor graph that are not part of the new decoding factor graph, starting from the root. He then applies the check-node and bit-node unitaries of the new decoding factor graph that are not yet applied and proceeds at its root in the same way.

After deciding $v_1$, Bob applies the inverse of every step above in reverse order, except the subtractions into $\bar R$. This returns the decision qudits and the eigen-list registers to their initial states. The phase decoder is therefore a unitary, followed by the subtractions into $\bar R$ and the inverse of the unitary. By the principle of deferred measurement, the decision qudits after the unitary hold the decisions of the successive-cancellation decoder with their correct probabilities. The phase decoder is therefore a coherent implementation, as in Lemma~\ref{lem:guessingfidelity}, of a projective measurement on $A'^mB^m$ and Bob's ancillas. With the ancillas in their initial state, the outcome of this measurement in $\bar R$ has the distribution of the outcome of $\Gamma$. Lemma~\ref{lem:guessingfidelity} then applies with the ancillas included in its system $B$.

In the real-register model, the eigen-list registers are infinite dimensional. For binary BPQM, a representation with finitely many bits per register approximates the decoder with an error that vanishes as the number of bits grows~\cite{piveteau_quantum_2022,piveteau_efficient_2025}. We expect this analysis to extend to $d$-ary eigen lists and leave it to future work.
The forward pass applies each check node and bit node unitary of the $\tfrac m2\log m$ kernel pairs once and inverts it once, so it uses $2m\log m$ node unitaries, including inverses, and the same number of eigen-list updates. A check node uses the fixed unitary $U^{\cnop}$ and an update with $O(d)$ arithmetic operations. A bit node uses $O(d^2)$ arithmetic operations for its update and for the amplitudes of the reflections in $U^{\vnop}_{\blambda_1,\blambda_2}$. These reflections use $O(d^2)$ two-level rotations controlled by the amplitudes.
Including the reverse pass, which doubles the number of steps, the phase decoder uses $O(d^2m\log m)$ operations, at most $2m$ eigen-list registers and $|\mathcal G_P|\leq m$ decision qudits at any time. The syndrome extraction $V$ and the final operation $W$ use $O(m\log m)$ controlled additions over $\ff_d$ since $H=\mathsf K_m^T$ has the recursive structure of the polar transform. Together with the amplitude decoder, the protocol uses $O(d^2m\log m)$ operations in total.

\section*{AI Statement}

The authors had been working on efficient quantum decoding constructions since well before the recent advances in AI, and investigating quantum communication problems to which BPQM can be applied. Based on several recent generalizations of BPQM, the authors independently identified PDPC channels as the class for which BPQM is optimal for entanglement generation. All results and proof ideas are the authors' own.

During the later stages of preparation, the authors used ChatGPT (GPT Sol-6 and earlier versions) and Claude (Opus 5.5) to check and refine the written presentation of the proofs, to improve the clarity of the manuscript, and to adapt our descriptions from previous works to the notation and setting of this paper. All AI-suggested changes were reviewed and verified by the authors. The authors take full responsibility for the content of the manuscript.

\printbibliography[heading=bibintoc,title={\fontsize{13pt}{16pt} References}]
\end{document}